\documentclass[pra,aps,superscriptaddress,twocolumn,notitlepage,letter,nopacs,nofootinbib,longbibliography]{revtex4-2}

\usepackage{graphicx,color,amsmath,amsfonts,enumerate,amsthm,amssymb,mathtools,enumitem,thmtools,hyperref,subfigure,mathdots,enumitem,centernot,bm,soul,bbm,stackrel}
\usepackage[capitalise, noabbrev]{cleveref}
\usepackage{mathrsfs}
\usepackage{babel}
\usepackage{stmaryrd}
\usepackage{multirow}
\hypersetup{colorlinks=true,linkcolor=blue,citecolor=blue,urlcolor=blue}

\usepackage{fontawesome}
 \usepackage{relsize}

\usepackage{tikz,pgfplots}
\usetikzlibrary{quantikz}

\def\D{ {\cal D} }
\def\E{ {\cal E} }

\def\H{ {\cal H} }
\def\I{ {\cal I} }

\def\L{ {\cal L} }

\def\N{ {\cal N} }

\def\U{ {\cal U} }
\def\V{ {\cal V} }

\def\bc{\mathbf{C}}
\def\>{\rangle}
\def\<{\langle}
\newcommand{\abs}[1]{\left| {#1} \right|}
\newcommand{\ketbra}[2]{\ensuremath{\left|#1\right\rangle\!\!\left\langle#2\right|}}
\renewcommand{\proj}[1]{\ensuremath{\left|#1\right\rangle\!\!\left\langle#1\right|}}

\newcommand{\Tr}{\mathrm{Tr}}
\newcommand{\fc}{\mathbf{C}}

\newcommand{\kkhide}[1]{}

\renewcommand{\v}[1]{\ensuremath{\boldsymbol #1}}

\definecolor{ppblue}{RGB}{46,117,182}
\definecolor{ppred}{RGB}{197, 90, 17}

\theoremstyle{plain}
\newtheorem{thm}{Theorem}
\newtheorem{lem}[thm]{Lemma}
\newtheorem{prop}{Proposition}

\theoremstyle{definition}
\newtheorem{defn}{Definition}

\begin{document}

\title{Dephasing Noise Simulation for Coherence-Generating Devices}

\author{Roberto Salazar}
\affiliation{Faculty of Physics, Astronomy and Applied Computer Science, Jagiellonian University, 30-348 Krak\'{o}w, Poland}
\affiliation{International  Centre  for  Theory  of  Quantum  Technologies,  University  of  Gdansk, Wita Stwosza 63, 80-308 Gdansk,  Poland}

\author{Fereshte Shahbeigi}
\affiliation{Faculty of Physics, Astronomy and Applied Computer Science, Jagiellonian University, 30-348 Krak\'{o}w, Poland}
\affiliation{RCQI, Institute of Physics, Slovak Academy of Sciences, D\'{u}bravsk\'{a} cesta 9, 84511 Bratislava, Slovakia}

\begin{abstract}
Advancing quantum technologies necessitates an in-depth exploration of how operations generate quantum resources and respond to noise. Crucial are gates generating quantum coherence and the challenge of mitigating gate dephasing noise. Precisely, we study the dephasing noise that reduces the coherence-generating power of quantum gates, its simulation, and critical factors. Our primary contribution lies in a theorem characterizing the full set of dephasing noises in gates, adaptable to the simulation by any predefined operation set. In particular, we apply our result to quantify the memory adaptability required for a dephasing noise to arise. Furthermore, we analytically calculate the quantifier for gates acting on qubit systems, thereby fully characterizing this scenario. Next, we show how our results reveal the structure of non-trivial dephasing noise affecting qubit gates and apply them to experimental data, conclusively demonstrating the existence of a gate's dephasing noise, which is irreducible to dephasing of either input or output states. Finally, we show how our study contributes to addressing an open question in the resource theory of coherence generation.
\end{abstract}

\maketitle
\section{Introduction}
\label{sec:intro}
\vspace{-0.4 cm}

A fundamental objective of quantum information is to study the quantum
properties of physical systems and demonstrate their advantages for
the development of new technologies \cite{QT1,QT2}. The investigation
and characterization of the practical benefits of multiple quantum
properties are currently carried out within the theoretical framework
of resource theories \cite{Gour2019,BG15,streltsov2017colloquium,streltsov2016quantum,TakagiRegula2019,Regula2020,Paul2019,Oszmaniec2019}. This class of theories allows organizing
quantum devices as resources according to their usefulness to exploit
a specific quantum property for practical purposes. Furthermore, resource theories allow the systematic study of the possible interconversions of such devices under procedures that do not increase the resource, denoted as free operations \cite{Gour2019}.
Channel resource theories are particularly relevant for quantum computing, as the branch that studies devices that transform quantum systems, such as quantum gates \cite{Carlo2020,Liu2020,Li2020,Winter2019}. In channel resource theories, the allowed free operations are superchannels whose implementation requires a quantum memory, pre- and post-processing channels \cite{KEYL2002}. Crucial is the theory of channels with the capacity to generate coherence due to the relevance, extensive research, and applications of the state's coherence \cite{streltsov2017colloquium,streltsov2016quantum,Zhang90,lostaglio2015description,lostaglio2015quantum,korzekwa2016extraction}. Moreover, recent advances show how coherence generation plays a fundamental role in the advantage provided by the famous Shor's algorithm \cite{Ahnefeld2022}.

However, quantum effects are inherently noise-sensitive, rendering any quantum advantage vulnerable to uncontrolled interactions with the environment \cite{Knill1997}. Thus, managing noise and decoherence effects is the main challenge in developing devices generating quantum resources. While experimental techniques play a crucial role, theoretical investigations also contribute significantly to progress in the field. One promising approach involves studying the mathematical structure of noise models to gain deeper insights into their properties and their impact on quantum resources \cite{Jiang2021,DephSuper2021,cai2023}.

In the context of devices generating a concrete resource, the essential noise to analyze would be solely degrading the resource's generation. Consequently, in the case of coherence-generating devices, we focus on dephasing noise over gates,  which affects the gate's action on coherences but does not affect the action on occupations.  Although dephasing noise models for states, known as dephasing channels, have been extensively researched and applied \cite{Breuer2007,wiseman2009}, their extension to gates is a recent development \cite{DephSuper2021,Daly2023}. This novel mathematical model, termed dephasing superchannels, was thoroughly elaborated upon in reference \cite{DephSuper2021}, offering a comprehensive understanding of its mathematical actions and physical implementation. Importantly, dephasing superchannels exhibit a distinct phenomenological novelty compared to dephasing channels \cite{DephSuper2021}, requiring further investigation.

In this work, our principal achievement lies in a formula that directly computes a pertinent Gram matrix from the simulation's operations and memory. The above Gram matrix identifies the simulated dephasing superchannel, and hence our formula simplifies the task of ascertaining whether a given set of operations and memory can simulate a specific noise. Moreover, the contrast with the simulations allows us to unequivocally identify the need for specific physical factors in the appearance of noise. This crucial observation suggests a compelling method to explore the phenomenology of dephasing superchannels, as illustrated in this work.

Our exploration of the formula's application starts investigating the memory activity required for simulating dephasing superchannels. Researchers have long recognized that the complexity of noise in gates arises from potential correlations between input and output systems, simulated through memory \cite{Chiribella2008}. However, one can model these correlations either using a memory with storage independent of the channel's input or employing a memory that dynamically adapts to such input state \cite{Megier2021,Elliot2022}. We denote the former as  \emph{passive} memory and the latter, as \emph{active} memory.

Employing our main theorem, we characterized the set of dephasing superchannels simulable by a passive memory and operations selected from arbitrary sets. Moreover, we introduced a quantifier of memory activity
based on the minimum distance between the Gram matrix identifying the target dephasing noise and those from the above set.

Contrary to expectations in \cite{DephSuper2021}, we uncovered the necessity of active memories, even when gates act on qubit systems. Concretely, we comprehensively determined dephasing noises due to passive memories for qubits and derived an analytical formula for the quantifier of memory activity. Furthermore, we applied our quantifier to analyze recent experimental data from nuclear magnetic resonance implementations, revealing an unprecedented phenomenology in dephasing quantum gate noise.

 Additionally, we exploit our results to address a stronger version of a long-standing conjecture in resource-generating theories, opening up new possibilities for resolving the main conjecture. Specifically, the conjecture \cite{Liu2020, Li2020, Winter2019} postulates that any free operation within the resource theory of coherence generation is implementable through memory and pre- and post-processing maximally incoherent operations (MIO) \cite{aberg2006superposition}. However, when we strengthen the conjecture and exclusively consider using passive memories in the implementation, our findings reveal dephasing superchannels that serve as compelling counterexamples. In addition, we outline a promising research direction that may address the most general case, a prospect we leave for future research.

The organization of this article is as follows: In Section \ref{sec:preliminary}, we
present the concepts of quantum mechanics which are relevant to our
work. Section \ref{sec:resource-theory} describes the formalism of channel resource theories.
In Section \ref{sec:dephasing-superchannel}, we introduce the class of dephasing superchannels.
Section \ref{sec:results} provides our main results, with an extensive presentation of the two applications. Finally, Section \ref{sec:discussion} discusses
the significance of our contribution, and outlines further research and
open questions.

\section{Preliminary concepts}
\label{sec:preliminary}
Our research is carried out using the standard Hilbert space formalism
for quantum systems, in which a preparation is represented by a positive
semidefinite operator $\rho$ of trace one, which belongs to the space
of bounded operators $\mathcal{B}\left(\mathcal{H}\right)$ of the
Hilbert space $\mathcal{H}$ and the allowed physical transformations
are represented by completely positive and trace preserving (CPTP)
maps \cite{KEYL2002,peres1995quantum}. In this article, we are mostly concerned
with finite degrees of freedom in which case the dimension of $\mathcal{H}$
is a natural number $d$ and operator $\rho$ reduces to a complex
density matrix of $d\times d$. The CPTP maps $\mathcal{E}$
are called  quantum channels and can be modeled in various ways, such
as the\emph{ Stinespring dilation} representation \cite{KEYL2002}:
\begin{align}
\mathcal{E}\left(\rho\right)=\mathrm{Tr}_{E}\left\{ U\left(\rho\otimes\left|0\right\rangle _{E}\!\left\langle 0\right|\right)U^{\dagger}\right\}
\end{align}
where $\mathcal{E}$ is realized by a unitary dynamics $U$
of an extended system followed by discarding the ancillary system
$E$. Another relevant representation for us is provided by the Choi-Jamio\l kowski
isomorphism \cite{choi1975completely,jamiolkowski1972linear} which leads to the \emph{Jamio\l kowski state} \cite{jamiolkowski1972linear}:
\begin{align}\label{eq:choi}
J\left(\mathcal{E}\right):=\mathcal{E}\otimes\mathcal{I}\left(\left|\Psi\right\rangle \!\left\langle \Psi\right|\right),\quad\left|\Psi\right\rangle:=\frac{1}{\sqrt{d}}\sum_{i=1}^{d}\left|ii\right\rangle
\end{align}
in which $\mathcal{I}$ is the identity channel and $\left|\Psi\right\rangle$ the maximally entangled state. Notably, the complete positivity of $\mathcal{E}$ is equivalent
to $J\left(\mathcal{E}\right)\geq0$ while the trace-preserving condition
gets mapped to $\mathrm{Tr}_{1}\left\{ J\left(\mathcal{E}\right)\right\} =\frac{1}{d}I$, with $I$ the identity matrix.

Moreover, for a given channel $\E(A\rightarrow B)$, with Kraus operator $\{K_n\}$ representation $\E(\rho)=\sum K_n\rho K_n^\dagger$, there exists the superoperator $\Phi^{\E}$ which is a  $d^2_B\times d^2_A$ dimensional matrix defined by
    \begin{align}
    \label{eq:superoperator}
        \Phi^{\E}_{\stackrel{\scriptstyle i j}{k l}}=\bra{i}\E\left(\ketbra{k}{l}\right)\ket{j}
    \end{align}
with a natural decomposition into Kraus operators  \cite{KarolZ2009}: $\Phi^{\E}=\sum_n K_n\otimes K_n^\ast$. The superoperator $\Phi^{\E}$ is a powerful technical tool, as it allows us to manipulate density matrices and channels in a similar way, as we do with vectors and matrices respectively. Additionally, $\Phi^{\E}$ is simply related to the  Jamio\l kowski state of the map by a special reordering of its entries called reshuffling \cite{KarolZ2009},
\begin{align}
    \label{eq:reshuffling}
    \Phi^{\E}_{\stackrel{\scriptstyle i j}{k l}}=d\times(J(\E)^R)_{\stackrel{\scriptstyle i j}{k l}}=d\times J(\E)_{\stackrel{\scriptstyle i k}{j l}}.
\end{align}


Other important notions for our investigation
are the set of incoherent states $\sigma=\sum_{i=0}^{d-1}p_{i}\left|i\right\rangle \!\left\langle i\right|$, with  $\{p_{i}\}$ a probability distribution over the preferred basis $\left\{ \left|i\right\rangle \right\} _{i=0}^{d-1}$ and the  set of  operations that leaves this set invariant known as \emph{maximally
incoherent operations} (MIO) \cite{aberg2006superposition}. Precisely, the later operations are the most
general class of resource non-generating operations in the resource
theory of coherence \cite{streltsov2017colloquium}.

A relevant subset of MIO is the so-called \emph{dephasing channels} \cite{li1997special}, which leave invariant the diagonal terms of any quantum state $\rho$. A dephasing channel $\mathcal{D}_{\mathrm{C}}$ is uniquely determined by a Gram matrix  $\mathrm{C}$, i.e. a positive matrix with all diagonal entries equal to one. Then, the action of $\mathcal{D}_{\mathrm{C}}$ on a quantum state $\rho$ is given by:
\begin{align}
\mathcal{D}_{\mathrm{C}}(\rho)=\rho\odot\mathrm{C}\label{eq:dephshur}
\end{align}

with $\odot$ a Schur product, which outputs a matrix with each element
the entry-wise multiplication of the two input matrices \cite{kye1995positive,li1997special}.
When $\mathrm{C}=I$, the identity matrix, the corresponding dephasing channel  $\mathcal{D}_{I}$ maps any state to a density matrix with the same diagonal, but zero coherence. The channel $\mathcal{D}_{I}$ denotes \emph{maximally dephasing channel} and has the following commutation relation with any MIO operation $\phi$ \cite{Gour2019}:

\begin{align}
\phi\circ\mathcal{D}_{I}= \mathcal{D}_{I }\circ\phi\circ\mathcal{D}_{I}\label{eq:miocomm}
\end{align}

Moreover, we can use $\mathcal{D}_{I}$ to write the specific action of a channel $\mathcal{E}$ over the diagonal terms of a state $\rho$ as:

\begin{align}
\mathcal{D}_{I }\circ\mathcal{E}\circ\mathcal{D}_{I}\label{eq:classact}
\end{align}

which we call the \emph{classical action} of the channel $\mathcal{E}$. A classical action is equivalent to a stochastic transition matrix $ T_{\mathcal{E}}$ acting on the diagonal terms of the density matrix of a state.
\section{Resource theories of Channels}
\label{sec:resource-theory}
Currently, the most fundamental quantum objects investigated in resource
theories consist of states $\rho$, measurements $\mathbf{M}$, and
quantum channels $\mathcal{E}$ \cite{Gour2019}. Of the three classes of objects mentioned,
the most complex are quantum channels, and consequently, their general
resource theoretical study has only recently started \cite{Liu2020,Li2020,Winter2019}. A channel
resource theory consists of the triple $\left\{ \mathfrak{F},\mathfrak{O},\Omega\right\} $,
where $\mathfrak{F}$ is a free set of CPTP maps, $\mathfrak{O}$
a set of free superchannels that map CPTP maps into CPTP maps and
$\Omega$ a measure of resourcefulness which maps CPTP maps into non-negative
numbers. As usual in the resource theoretical framework the superchannels
$\mathfrak{O}$ should map elements of the free set $\mathfrak{F}$ into itself,
the measure $\Omega$ must be zero for any free channel $\mathcal{T}\in\mathfrak{F}$
and it should not increase under any free superchannel $\Xi\in\mathfrak{O}$,
i.e. for all CPTP maps $\mathcal{E}$,
we have:
\begin{align}
\Omega\left(\Xi\left[\mathcal{E}\right]\right)\leq\Omega\left(\mathcal{E}\right).
\end{align}
Additionally, it is usual that free operations $\mathfrak{O}$ include
left or right composition $\circ$ and tensoring $\otimes$ under
free maps $\mathcal{T}\in\mathfrak{F}$ with the identity map \cite{Winter2019}.

A reason for the complexity and wide phenomenology associated with channel resource theories lies in the structure of their free operations
$\mathfrak{O}.$ In the most general sense, any physically realizable superchannel $\Xi$
is composed of three elements: memory systems $D^{\prime\prime},D^{\prime}$, and $D$, an encoding map $\mathcal{N}_{en}\left(A^{\prime}\otimes D^{\prime}\rightarrow A\otimes D\right)$, and
a decoding map $\mathcal{N}_{de}\left(B\otimes D\rightarrow B^{\prime}\otimes D^{\prime\prime}\right)$
 followed by tracing out the memory, which are applied on a channel $\mathcal{E}\left(A\rightarrow B\right)$
in the following way:
\begin{align}
\!\Xi\left[\mathcal{E}\right](\rho_{A^{\prime}})\!=\! \Tr_{D^{\prime\prime}}\!\circ\mathcal{N}_{de}\circ\left(\mathcal{E}\otimes\mathcal{I}_{D}\right)\circ\mathcal{N}_{en}(\rho_{A^{\prime}}\otimes\tau_{D^{\prime}}) \label{eq:superchannel}
\end{align}
giving the channel $\Xi\left[\mathcal{E}\right]\left(A^{\prime}\rightarrow B^{\prime}\right)$. Diagrammatically,

\begin{equation}
		\label{eq:diagram1}
		\begin{quantikz}
			& \gate{\Xi[\E]} &\qw
		\end{quantikz}
		=
		\begin{quantikz}
			&[-0.1cm]	\gate[wires=2]{\mathcal{N}_{en}}&[-0.1cm] \gate{\E} &[-0.1cm] \gate[wires=2]{\mathcal{N}_{de}}&[-0.6cm]\qw \\[-0.5cm]
				& 	 &\qw  	& 	& [-0.6cm]	\trash{\text{\emph{discard}}}
		\end{quantikz}\!\!\!.
\end{equation}

Due to the generality, non-uniqueness, and complexity of representation
(\ref{eq:superchannel}), it became a widely accepted approach to construct channel resource theories based on an underlying state
resource theory \cite{Liu2020,Li2020, Winter2019,Saxena2020,Chen2020entanglement}. Let us define $\left\{ F_{s},O_{s},M_{s}\right\} $ respectively as the triad of free states, free operations, and monotonic measure of the state resource theory.
The method mentioned above focuses on the potential of a channel $\mathcal{E}$
to generate or increase state resources and build up the free superchannels
$\Xi$ using elements from the associated state resource theory. The previous approach restricts the initial state of the ancillary system
$D$ in (\ref{eq:superchannel}) to be a free state $\tau_{D^{\prime}}\in F_{s}$,
and both encoding and decoding maps
to be free operations, i.e  $\mathcal{N}_{en},\, \mathcal{N}_{de}\in O_{s}$. It is straightforward to
check that superchannels satisfying the preceding conditions preserve
free channels, and due to their constructive nature, were denoted
as \emph{freely realizable} operations in \cite{Carlo2021}.

\section{Dephasing superchannels}
\label{sec:dephasing-superchannel}
In reference \cite{DephSuper2021} it was introduced a class
of superchannels denoted as \emph{dephasing superchannels} which
models dephasing noise on gates. The formal definition of a dephasing superchannel is as follows:

\begin{defn}[Dephasing superchannel]
    A quantum superchannel $\Xi$ is called a dephasing superchannel if the transition probabilities in the distinguished basis are invariant under $\Xi$:
	\begin{align}
    	\forall~\E,\ket{i},\ket{j}:\quad \bra{i}\Xi[\E](\ketbra{j}{j})\ket{i}=\bra{i}\E(\ketbra{j}{j})\ket{i}.
	\end{align}	
\end{defn}

Essentially, a dephasing superchannel leaves invariant the classical action of any channel $\E$:
\begin{align} \label{eq:invariant-classical}
\mathcal{D}_{I}\circ\Xi\left(\mathcal{E}\right)\circ\mathcal{D}_{I}=\mathcal{D}_{I}\circ\mathcal{E}\circ\mathcal{D}_{I}
 \end{align}

Analogously to dephasing channels, a dephasing superchannel is uniquely determined by a Gram matrix. From \cite{DephSuper2021} we know that the Gram matrix $\mathbf{C}$ identifying a dephasing superchannel $\Xi_{\mathbf{C}}$  has a further structure:

	\begin{align}
	\label{eq:corr_form}
	\mathbf{C}=\left[\begin{array}{cccc}
	{C_{00}} & {C_{01}} & {\ldots} & {C_{0\, d-1}} \\
	{C_{10}} & {C_{00}} & {\ldots} & {C_{1\, d-1}} \\
	{\vdots} & {\vdots} & {\ddots} & {\vdots} \\
	{C_{d-1\, 0}} & {C_{d-1\,  1}} & {\ldots} & {C_{00}}
	\end{array}\right],
	\end{align}
	where $C_{ij}$ are $d\times d$ matrices and $C_{00}$ is a Gram matrix itself.

Additionally, a dephasing superchannel $\Xi_{\mathbf{C}}$
has a physical realization as unitary pre- and post-processing with
an ancillary system of dimension $d$ as follows:
\begin{align}
\Xi_{\mathbf{C}}\left[\mathcal{E}\right]\left(\rho\right)=\mathrm{Tr}_{D}\left\{ \mathcal{V}\circ\left[\mathcal{E}\otimes\mathcal{I}_{D}\right]\circ\mathcal{U}\left(\rho_{A}\otimes\left|0\right\rangle _{D}\!\left\langle 0\right|\right)\right\}
\end{align}
or as a circuit diagram:
\begin{equation}
		\label{eq:diagram}
		\begin{quantikz}
			& \gate{\Xi_C[\E]} &\qw
		\end{quantikz}
		=
		\begin{quantikz}
			&[-0.1cm]	\gate[wires=2]{\U}&[-0.1cm] \gate{\E} &[-0.1cm] \gate[wires=2]{\V}&[-0.6cm]\qw \\[-0.5cm]
			\ketbra{0}{0}	&  	 & \qw 	&	& [-0.6cm]	\trash{\text{\emph{discard}}}
		\end{quantikz}\!\!\!,
	\end{equation}
where,
\begin{eqnarray}
\label{eq:CU}
\mathcal{U}\left(\cdot\right)=U\left(\cdot\right)U^{\dagger}, \  &  & U=\sum_{i=0}^{d-1}\left|i\right\rangle \!\left\langle i\right|\otimes U_{i}\\
\label{eq:CV}
\mathcal{V}\left(\cdot\right)=V\left(\cdot\right)V^{\dagger},\  &  & V=\sum_{i=0}^{d-1}\left|i\right\rangle \!\left\langle i\right|\otimes V_{i}
\end{eqnarray}
with $\left\{ U_{i}\right\},\,\left\{ V_{i}\right\} $ being
 unitaries of size $d^{2}$. The unitaries define component-wise the
 Gram matrix $\mathbf{C}$:
\begin{align}
\mathbf{C}_{ik,jl}=\left\langle 0\right|U_{l}^{\dagger}V_{j}^{\dagger}V_{i}U_{k}\left|0\right\rangle \label{eq:Scorrematrix}
\end{align}
from (\ref{eq:Scorrematrix}) follows that $\mathbf{C}$  is the Gram matrix of the pure states $\left|\psi_{ik}\right\rangle=V_{i}U_{k}\left|0\right\rangle$ generated by matrices $V_{i},U_{k}$.

Moreover, the Gram matrix $\mathbf{C}$
provides an elegant relation between the Jamio\l kowski states $J\left(\mathcal{E}\right)$
and $J\left(\Xi_{\mathbf{C}}\left[\mathcal{E}\right]\right)$ \cite{DephSuper2021}:
\begin{align}
J\left(\Xi_{\mathbf{C}}\left[\mathcal{E}\right]\right)=J\left(\mathcal{E}\right)\odot\mathbf{C}\label{eq:jamshur}
\end{align}
with $\odot$ a Schur product. Crucially,
in reference \cite{DephSuper2021} it was shown that for any resource monotone
of coherence $\mathscr{C}$ in a preferred basis $\left\{ \left|k\right\rangle \right\} $
the corresponding coherence generating power $\mathscr{C}_{g}$ of
a channel $\mathcal{E}$ \cite{mani2015cohering}:
\begin{align}
\mathscr{C}_{g}\left(\mathcal{E}\right)=\max_{k}\left\{ \mathscr{C}\left(\mathcal{E}\left(\left|k\right\rangle \!\left\langle k\right|\right)\right)\right\}
\end{align}
is non-increasing under any dephasing superchannel $\Xi_{\mathbf{C}}$:
\begin{align}
\mathscr{C}_{g}\left(\Xi_{\mathbf{C}}\left[\mathcal{E}\right]\right)\leq\mathscr{C}_{g}\left(\mathcal{E}\right) \label{dephmon}
\end{align}
Furthermore, all $\Xi_{\mathbf{C}}$ by definition  commute with the maximally dephasing superchannel --which has a Gram matrix $\mathbf{C}=\mathbf{I}$, the identity matrix-- and in consequence from Theorem 1 of reference \cite{Saxena2020} follows that $\mathcal{I}\!\otimes\!\Xi_{\mathbf{C}}$ does not increase the coherence generation, i.e. $\Xi_{\mathbf{C}}$ is \emph{completely resource non increasing} \cite{Saxena2020}.
Because $\mathscr{C}_{g}$ is the natural monotone $\Omega$ for the resource theory of coherence generation \cite{Li2020},  property (\ref{dephmon}) makes the superchannels $\Xi_{\mathbf{C}}$ potential
free operations of a resource theory of coherence generation.

\section{Results}
\label{sec:results}

\subsection{Simulation of a dephasing superchannel from general pre- and post-processing }
\label{subsec:simulation}

This section presents our principal achievement encompassing two pivotal theorems. Together, these theorems determine the complete set of simulations (\ref{eq:superchannel}) for dephasing superchannels, spanning all possible scenarios. The first theorem unveils a direct formula for the Gram matrix $\mathbf{C}$, derived from the superoperators associated with the pre- and post-processing operations and the initial memory state. Subsequently, the second theorem serves as an indispensable complement, presenting the necessary and sufficient conditions under which the operations employed in (\ref{eq:superchannel}) can simulate a dephasing superchannel. Importantly, given the unique identification of the dephasing superchannel by matrix $\mathbf{C}$, Theorems \ref{thm:gram-pre-post} and \ref{thm:gram-pre-post2} offer versatile technical tools adaptable to address any inquiry regarding dephasing noise in gate simulations.

We start with our main result, consisting of the greatest possible generalization of the simulation introduced in Ref. \cite{DephSuper2021}, and presented in  Eq.~\eqref{eq:Scorrematrix}, where the encoding and the decoding maps are assumed to be controlled-unitary gates.
\begin{thm}
\label{thm:gram-pre-post}
    Let $\Phi^{\mathsmaller{en}}$ and $\Phi^{\mathsmaller{de}}$ denote the superoperator representations, introduced in Eq.~\eqref{eq:superoperator}, of $\N_{en}$ and $\N_{de}$. For a dephasing superchannel $\Xi_\mathbf{C}$ corresponding to the Gram matrix $\mathbf{C}$, one has
    \begin{align}
    \label{eq:gram-pre-post}
    \mathbf{C}_{\stackrel{\scriptstyle i k}{j l}}\delta_{ip}\delta_{jq}\delta_{km}\delta_{ln}=\sum_{\stackrel{\scriptstyle \theta,\gamma,\eta}{\alpha,\beta}}\Phi^{\mathsmaller{de}}_{\stackrel{\scriptstyle i\theta j\theta}{p\gamma q\eta}}\ \Phi^{\mathsmaller{en}}_{\stackrel{\scriptstyle k\gamma l\eta}{m\alpha n\beta}}\tau_{\alpha,\beta},
    \end{align}
    or in terms of the encoding and decoding maps, we have equivalently

    \begin{align}
    \mathbf{C}_{\stackrel{\scriptstyle i k}{j l}}\delta_{ip}\delta_{jq}\delta_{km}\delta_{ln}=&\sum_{\gamma,\eta}\bra{i}\Tr_2\left[\N_{de}\left(\ketbra{p}{q}\otimes\ketbra{\gamma}{\eta}\right)\right]\ket{j} \nonumber\\
    & \times \bra{k\gamma}\N_{en}\left(\ketbra{m}{n}\otimes\tau\right)\ket{l\eta}
    \end{align}
    where $\tau$ is the initial state of the memory and we used the Latin letters for the system degrees of freedom and the  Greek letters for the memory ones.
\end{thm}

We prove Theorem~\ref{thm:gram-pre-post} in Appendix~\ref{app:CBasedOnMaps}. Additionally, we emphasize for future reference that if the initial state is  classical, i.e., $\tau=\sum p_\alpha\proj\alpha$, the equation on the Gram matrix is reduced to:
    \begin{align}
    \label{eq:gram-pre-post1}
    \mathbf{C}_{\stackrel{\scriptstyle i k}{j l}}=\sum_{\stackrel{\scriptstyle\theta,\gamma,\eta}{\alpha}} p_\alpha\ \Phi^{\mathsmaller{de}}_{\stackrel{\scriptstyle i\theta j\theta}{i\gamma j\eta}}\ \Phi^{\mathsmaller{en}}_{\stackrel{\scriptstyle k\gamma l\eta}{k\alpha l\alpha}}.
    \end{align}

The general and compact nature of formulae (\ref{eq:gram-pre-post}), (\ref{eq:gram-pre-post1}) allows us to directly apply our knowledge of the operations to determine the simulable dephasing superchannels. Moreover, given the freedom of the set of operations $\N_{en}$ and $\N_{de}$, we can opt for one or more constraints to define the simulable noise (expressed through Gram matrices $\mathbf{C}$). Consequently, any noise beyond this specified set signifies a breach of the predetermined constraints. Harnessing this straightforward observation, we can utilize equations (\ref{eq:gram-pre-post}) and (\ref{eq:gram-pre-post1}) to identify essential physical factors within a noise. Furthermore, we will exemplify how to determine which noises pose greater challenges in approximation due to the physical factor encapsulated by the constraints.

The strength of Theorem \ref{thm:gram-pre-post} rests in its applicability to any set of $\N_{en}$, $\N_{de}$, operations simulating a dephasing superchannel. Because of the above, it is crucial to determine under what conditions such encoding and decoding maps simulate a dephasing superchannel.
We move in this direction by demonstrating in Appendix~\ref{app:dephasing-on-system} that a necessary condition on the encoding and decoding maps realizing a dephasing superchannel is to respect the following restrictions for any $\rho$ and $m$:
\begin{subequations}\label{eq:N-pre-post}
\begin{align}
    &\!\!\Tr_2\!\left[\N_{en}\!\left(\D_I(\rho)\!\otimes\!\tau\right)\right]\!=\!\D_I\!\left(\Tr_2\!\left[\N_{en}(\D_I(\rho)\!\otimes\!\tau)\right]\right)\!,\label{subeq:NS-pre}\\
    &\!\!\D_I\!\left(\Tr_2\!\left[\N_{de}(\rho\otimes\sigma_{m})\right]\right)\!=\!\D_I\!\left(\Tr_2\!\left[\N_{de}\left(\D_I(\rho)\otimes\sigma_{m}\right)\right]\right)\!,\label{subeq:NS-post}
\end{align}
\end{subequations}
where $\D_I$ is the completely dephasing channel \eqref{eq:dephshur} acting on the system, and
\begin{align}\label{eq:sigma}
\sigma_{m}:=\Tr_1\left[\N_{en}\left(\proj m\otimes\tau\right)\right].
\end{align}
The above shows that, on the system degrees of freedom, the pre-processing has to be a maximally incoherent operation \eqref{eq:miocomm} while the post-processing map cannot use the coherence of the input to affect the population of the output.  This, however, is not a sufficient condition. More precisely,  we prove in the following theorem that indeed they have to act like a dephasing channel on the system (Proof in Appendix~\ref{app:proof-of-T2}):
\begin{thm}
\label{thm:gram-pre-post2}
The encoding and decoding maps realize a dephasing superchannel if and only if their effect on the system degrees of freedom is dephasing. This means that for any $\rho$ and $m$
\begin{subequations}\label{eq:NS-enc-dec}
\begin{align}
   \Tr_{2}\left[\N_{en}\left(\rho\otimes\tau\right)\right]&=\rho\odot\mathrm{C}_{en},\label{subeq:enc-deph}\\
    \Tr_{2}\left[\N_{de}\left(\rho\otimes\sigma_{m}\right)\right]&=\rho\odot\mathrm{C}_{de}^{m}\label{subeq:dec-deph},
\end{align}
\end{subequations}
where $\tau$ is the initial state of the environment, and $\sigma_{m}$ is given in Eq.~\eqref{eq:sigma}. Moreover, the Gram matrices $\mathrm{C}_{en}$ and $\mathrm{C}_{de}^m$ for any $m$ are the marginals of $\mathbf{C}$, the Gram matrix of the dephasing superchannel (see Eq.~\eqref{eq:corr_form}), given by
\begin{align}\label{eq:marginal-gram}
    \mathrm{C}_{en}=C_{00},\quad\mathrm{C}_{de}^m=\sum_{ij} \left(C_{ij}\right)_{mm}\ketbra{i}{j}.
\end{align}
\end{thm}

The conditions presented in Theorem \ref{thm:gram-pre-post2} complement Theorem \ref{thm:gram-pre-post} since they provide the intrinsic limits in the freedom of the simulation. Furthermore, conditions (\ref{subeq:enc-deph}), (\ref{subeq:dec-deph}), and (\ref{eq:marginal-gram}) make explicit the marginal action of the encoding and decoding maps over the system, using the terms from the Gram matrix  $\mathbf{C}$ directly. The above characterization of the marginal actions over the system contributes in a clear way to fix the parameters of operations $\N_{en}$, $\N_{de}$ from the parameters defining the dephasing superchannel. Indeed, the previous feature of Theorem \ref{thm:gram-pre-post2} will be crucial in transforming intricate general simulation problems into tractable ones.

\subsection{Quantifying the required active memory}
\label{subsec:memory}

A memory device's fundamental faculty is storing information to influence future action. The simplest way to use a memory device to implement a quantum process is by coding a program that determines the operations over the system through correlations. Such a capacity of a memory device is an essential resource in quantum processes \cite{Rosset2018,Buscemi2020,XiaoRegula2021}. However, the quantum interactions of a system with memory can generate effects beyond programming. For example, quantum processes intended for computation can drastically reduce the size of the system used by including a quantum memory, allowing a more scalable computational architecture \cite{Gouzien2021,Zhang2022}.

As the previous applications demand a quantum memory that dynamically adapts to the input of the process, it becomes relevant to distinguish those processes that require such kind of memory from those achievable through predefined programming. The latter kind of memory is characterized by operating without storing information from the input state, and we designate them as \emph{passive} memories, while the former, in contrast, we call \emph{active} memories.
The previous concepts can be formalized in the following way:

\begin{defn}
\label{defn:passive}
    Whenever a superchannel implemented as in (\ref{eq:superchannel}), has a memory:
    \begin{align}
    \label{eq:passive}
        \sigma:=\Tr_1\left[\N_{en}\left(\varrho\otimes\tau\right)\right],
    \end{align}
    such that $\sigma$ is always independent of the system's input $\varrho$, we say the superchannel's memory is \emph{passive} and otherwise \emph{active}.
    \end{defn}

Notably, section III.C of Ref. \cite{DephSuper2021} presents an example establishing the phenomenological novelty of dephasing superchannels as it shows that some are irreducible to noise models based solely on dephasing channels. In this section, we go one step further, exploiting Theorems \ref{thm:gram-pre-post} and \ref{thm:gram-pre-post2} to show the  essential role of active memory in the new phenomenology and to assess it quantitatively.

More precisely, for a passive memory, the state $\sigma_{m}$ in Eq.~\eqref{eq:sigma} is independent of $m$, which implies the right-hand-side of Eq.~\eqref{subeq:dec-deph} should also be $m$-independent.  Additionally, through Eq.~\eqref{eq:marginal-gram}, follows that:
\begin{align}
\label{eq:diagonal-entries}
   \mathbf{C}_{\stackrel{\scriptstyle i m}{j m}}= (C_{ij})_{mm},
\end{align}
does not depend on $m$. Thus, the diagonal entries of the block $ij$ of the correlation matrix $\mathbf{C}$ corresponding to the dephasing superchannel, given by Eq.~\eqref{eq:corr_form},  are fixed.    The inverse of this is not necessarily true. To see this, notice that it is possible to have the left-hand side of Eq.~\eqref{subeq:dec-deph} $m$-dependent, therefore having an active memory, yet map it (by a semi-causal channel \cite{Beckman-Causality} for example) to an $m$-independent state on the right-hand side of this equation. However, in the following subsection, we will show for dephasing superchannnels acting on qubit maps the two sets coincide. Furthermore, we emphasize the assumption that $\sigma_{m}=\sigma$ is independent of $m$ does not mean that the system evolution is independent of its environment, see Appendix~\ref{app:off-diagonal} for more details.


 Being $m$-independent, Eq.\eqref{eq:diagonal-entries} demonstrates that the superchannel realized by a passive memory forms only a measure zero subset of the entire set of dephasing superchannels. Therefore, one can associate the dephasing superchannels achievable by a passive memory with a limited set of Gram matrices and perform the quantitative evaluation of an active memory using a distance measure from that set:
\begin{defn}
\label{defn:measure}
    A quantification of the amount of active memory $M(\Xi_{\fc})$ needed to realize a dephasing superchannel, isomorphic to the Gram matrix $\fc$, is:
    \begin{align}
    \label{eq:measure}
        M(\Xi_{\fc})=\min_{\fc_\star\in\mathcal{S}}D(\fc,\fc_\star),
    \end{align}
    where $D$ is any pseudo distance \cite{bengtsson2017,GourWinter2019} and the optimization is over $\mathcal{S}$, the set of all correlation matrices obtained by a passive memory.
    \end{defn}

    Later, we will exploit the above definition to identify promising dephasing superchannels for experimentally testing non-trivial noises on quantum gates.

\subsection{Dephasing superchannels acting on qubit channels}
\label{subsec:qubit}
The point of this section is to study the active memory needed in the realization of a dephasing superchannel acting on qubit channels. For qubit channels, we first show that the set of dephasing superchannels achievable through passive memory follows a specific pattern: $\fc=\sum_i q_{i} \mathrm{C}^{(i)}_1\otimes \mathrm{C}^{(i)}_2$. This can be understood as the application of $\D_{C_1^{(i)}}\!\circ\E\circ\D_{C_2^{(i)}}$ with certain classical probabilities $q_i$, as explained in \cite{DephSuper2021}, rather than using more complex higher-dimensional encoding-decoding maps. We later show for a generic dephasing superchannel the amount of active memory is quantified by the difference between two diagonal elements of the off-diagonal block $C_{01}$.

To convey the central result of this section, we adopt the following notation to represent the  $l_1$ distance of two matrices $G$ and $H$:
\begin{equation}\label{eq:l-1-dis}
    D_{l_1}(G,H)=\sum_{i,j}|G_{ij}-H_{ij}|.
\end{equation}
Then, we will apply the following result from Ref.~\cite{DephSuper2021}:
\begin{lem}
\label{lem:gram-rank-2}
    Any $4\times4$ Gram matrix $\fc$ corresponding to a dephasing superchannel $\Xi_\fc$, which acts on qubit channels, is proportional to a two-qubit separable state.
\end{lem}
Furthermore, we exploit the fact that any trace-preserving linear transformation, denoted by  $\L$,  acting on quantum states corresponds to an affine transformation applied to their Bloch vectors. This correspondence is established by
\begin{equation}
    \!\!\rho=\frac{1}{2}\left(I+\mathbf{r}.\boldsymbol{\sigma}\right)\!\!\quad{\stackrel{\scriptstyle \L}{\longrightarrow}}\quad\!\! \L(\rho)=\frac{1}{2}\Big(I+(\Lambda\mathbf{r}+\mathbf{t}).\boldsymbol{\sigma}\Big),
\end{equation}
where $\Lambda$ is a $3\times 3$ distortion matrix and $\mathbf{t}$ is a translation vector in $\mathbb{R}^3$ showing how $\L$ shifts $I$. If $\L$ is in addition positive, then $\Lambda$ has to be a contraction, i.e., $\Lambda^\dagger \Lambda\leq I$. In particular,
\begin{equation}
\label{eq:positive-matrix}
    \Lambda_\star=
    \begin{pmatrix}
       1&0&0\\
       0&1&0\\
       0&0&0
    \end{pmatrix},\quad
    \mathbf{t}_\star=
    \begin{pmatrix}
        0\\0\\0
    \end{pmatrix},
\end{equation}
define the positive linear transformation $\L_\star$ which is not completely positive and projects any point in the Bloch sphere into $x-y$ plane. Now, we bring one of the main results of this section.

\begin{prop}
\label{prop:qubit-passive-memry}
Let  $\Xi_{\fc}$ be a dephasing superchannel acting on qubit channels with $\fc$ being its associated $4\times4$ Gram matrix, i.e.,
\begin{align}
\label{eq:4dim-gram-block}
    \fc=\left(\begin{array}{@{}cc@{}}
  \begin{matrix}
     C_{00}
  \end{matrix}
  & C_{01} \\
   C_{01}^\dagger &
  \begin{matrix}
  C_{00}
  \end{matrix}
\end{array}\right).
\end{align}
Then, the following are equivalent:
\begin{enumerate}[label=(\roman*)]
    \item $\Xi_\fc$ can be realized employing some passive memory.
    \item $\fc=\sum q_{i} \mathrm{C}^{(i)}_1\otimes \mathrm{C}^{(i)}_2$.
    \item $C_{01}$ has a fixed diagonal entry.
\end{enumerate}
\end{prop}
\begin{proof}
The fact that $(ii)$ is achievable by applying dephasing channels, according to some classical memory, before and after $\E$ shows that $(ii)$ gives $(i)$. It is also clear that $(i)$ implies $(iii)$ as discussed in the last subsection. To complete the proof, it remains to show $(iii)$ results in $(ii)$. In that order, we will apply Lemma~\ref{lem:gram-rank-2}  to show any $4\times4$ Gram matrix of the above form with fixed diagonal entries of $C_{01}$ can be written as $\sum q_{i} \mathrm{C}^{(i)}_1\otimes \mathrm{C}^{(i)}_2$.

In that order, we start by noticing that since $\fc$ is separable  by Lemma~\ref{lem:gram-rank-2}, it admits the following decomposition
\begin{equation}
\label{eq:sep-dec}
    \fc=4\sum p_i\proj{\psi_i}\otimes\proj{\phi_i}.
\end{equation}
Let $\mathbf{r}^i=(r_1^i,r_2^i,r_3^i)^\mathrm{T}$ and $\mathbf{s}^i=(s_1^i,s_2^i,s_3^i)^\mathrm{T}$ respectively denote the Bloch vectors of $\proj{\psi_i}$ and $\proj{\phi_i}$, i.e.,
\begin{equation}
\label{eq:Bloch}
    \proj{\psi_i}=\frac{1}{2}\left(I+\mathbf{r}^i.\boldsymbol{\sigma}\right),\quad \proj{\phi_i}=\frac{1}{2}\left(I+\mathbf{s}^i.\boldsymbol{\sigma}\right).
\end{equation}
Now, inserting Eq.~\eqref{eq:Bloch} in Eq.~\eqref{eq:sep-dec},  it is straightforward to see it gives a Gram matrix of the form~\eqref{eq:4dim-gram-block} if and only if all the following conditions hold
\begin{subequations}\label{eq:N-on-dec}
\begin{align}
    &\sum p_i r^i_3=0\\
    &\sum p_i s^i_3=0\\
    &\sum p_i r^i_3 s^i_j=0 \quad\forall j\in\{1,2,3\}.
\end{align}
\end{subequations}

Next, applying the condition that the diagonal element of $C_{01}$ is fixed, it must hold that
\begin{align}\label{eq:S-on-dec}
    \sum p_i r^i_j s^i_3=0 \quad\forall j\in\{1,2\}.
\end{align}

As a result of Eqs.~\eqref{eq:N-on-dec} together with Eq.~\eqref{eq:S-on-dec}, the $z$ components of the Bloch vectors $\mathbf{r}_i$ and $\mathbf{s}_i$ vanish from the Gram matrix $\fc$ in Eq.~\eqref{eq:4dim-gram-block}. Therefore, $\fc$ remains invariant under the action of the map $\L_\star\otimes\L_\star$, where affine representation of $\L_\star$ is given by Eq.~\eqref{eq:positive-matrix}, i.e.,
\begin{equation}
    \fc=4\sum p_i\L_\star(\proj{\psi_i})\otimes\L_\star(\proj{\phi_i}).
\end{equation}
Noting that $\L_\star$ maps all the qubit states to the $x-y$ plane and the fact that any state in this plane is proportional to a $2\times2$ Gram matrix completes the proof.
\end{proof}
The above proposition proves that the noise imposed by a passive memory in the qubit case is realizable by applying some dephasing channels as pre- and post-processing chosen according to some classical memory. It remains open to prove if the same holds in higher dimensions as well, i.e., whether a passive memory in any dimension is equivalent to a classical memory.
Furthermore, in the qubit case where $\fc$ is separable, the presence of a fixed diagonal within the off-diagonal block suggests the possibility of realization through a passive memory. Examining whether separability, combined with these fixed entries, implies the existence of a passive memory in higher dimensions, is an avenue for future research.

On the other hand, the difference between the diagonal elements of $C_{01}$, the off-diagonal block of a $4\times4$ Gram matrix $\fc$, can show the amount of active memory required in the realization of a dephasing superchannel acting on qubit channels. In what follows, we prove this intuition. More precisely, we show, by assuming the $l_1$ distance \eqref{eq:l-1-dis}, the modulus of this difference between these two elements quantifies the memory through Eq.~\eqref{eq:measure}.
\begin{thm}
\label{thm:qubit-memory}
    The minimum amount, quantified by $l_1$ distance, of the active memory needed in the realization of a dephasing superchannel $\Xi_\fc$, which acts on qubit channels and is isomorphic to the Gram matrix $\fc$,  is given by
    \begin{equation}\label{eq:qubit-memory}
        M(\Xi_\fc)=\min_{\fc_\star\in S}D_{l_1}(\fc,\fc_\star)=2|\fc_{\stackrel{\scriptstyle 00}{10}}-\fc_{\stackrel{\scriptstyle 01}{11}}|,
    \end{equation}
    where $S$ is the set of Gram matrices realizable without an active memory.
\end{thm}
\begin{proof}
    To prove, assume that $\min$ in Eq.~\eqref{eq:qubit-memory} is obtained by the Gram matrix $\fc'$. Thus, we get
    \begin{align}\label{eq:lower-bound}
        M(\Xi_\fc)&=D_{l_1}(\fc,\fc')=\sum_{i,j,k,l}|\fc_{\stackrel{\scriptstyle ik}{jl}}-\fc'_{\stackrel{\scriptstyle ik}{jl}}|\nonumber\\
        &\geq2\left(|\fc_{\stackrel{\scriptstyle 00}{10}}-\fc'_{\stackrel{\scriptstyle 00}{10}}|+|\fc_{\stackrel{\scriptstyle 01}{11}}-\fc'_{\stackrel{\scriptstyle 01}{11}}|\right)\nonumber\\
        &=2\left(|\fc_{\stackrel{\scriptstyle 00}{10}}-\fc'_{\stackrel{\scriptstyle 00}{10}}|+|\fc_{\stackrel{\scriptstyle 01}{11}}-\fc'_{\stackrel{\scriptstyle 00}{10}}|\right)\nonumber\\
        &\geq2|\fc_{\stackrel{\scriptstyle 00}{10}}-\fc_{\stackrel{\scriptstyle 01}{11}}|,
    \end{align}
    where the first inequality holds due to the Hermiticity of a Gram matrix and because of taking $4$ out of $16$ non-negative elements of the summation. The following equality is obtained by the fact that $\fc'$ has a passive memory, and thus has the same diagonal element on its off-diagonal block. Finally, the last inequality is due to the triangle inequality.

    Having the lower bound~\eqref{eq:lower-bound}, to complete the proof of the theorem, it is enough to show that there exists a Gram matrix $\fc'\in S$ that satisfies this bound with equality. It is easy to show such a $\fc'$ can be obtained by applying $\L_\star$ \eqref{eq:positive-matrix} locally on $\fc$
    \begin{equation}
        \fc'=(\I\otimes\L_\star)(\fc).
    \end{equation}
    The facts  that $\L_\star$ is a positive map and $\fc$ is separable imply that the above gives a positive matrix, consequently, a valid Gram matrix in $S$ and this completes the proof.
\end{proof}

From the above result, we can find simple realizations of dephasing superchannels maximally departing from those simulable with dephasing channels and passive memory. An example of such a Gram matrix is given by \begin{align}\label{eq:maxEx1}
	\fc_{max}=\begin{pmatrix}
			1&0&1&0\\
			0&1&0&0\\
			1&0&1&0\\
			0&0&0&1
		\end{pmatrix}.
\end{align}

To understand the previous case in more detail, we can study its action in the Bloch sphere. Writing the affine transformation of a general qubit quantum channel $\L$ as
\begin{equation}
\label{eq:positive-matrix-general}
    \Lambda=
    \begin{pmatrix}
       u_{1} & v_{1} & w_{1}\\
 u_{2} & v_{2} & w_{2}\\
 u_{3} & v_{3} & w_{3}
    \end{pmatrix},\quad
    \mathbf{t}=
    \begin{pmatrix}
        t_1\\t_2\\t_3
    \end{pmatrix},
\end{equation}
then, exploiting definition (\ref{eq:choi}), and the identity \cite{Bertlmann2008,Gamel2016}:

\begin{equation}
\left|\Psi\right\rangle \!\left\langle \Psi\right|=\frac{1}{4}(I\otimes I+\sigma_{1}\otimes\sigma_{1}-\sigma_{2}\otimes\sigma_{2}+\sigma_{3}\otimes\sigma_{3}),
\end{equation}
by applying $\L$ to the first system we obtain
the corresponding Jamio\l kowski state $J(\L)$:
\begin{widetext}

\begin{equation}
\label{eq:jam-general}
J(\L)=\frac{1}{4}\left[\begin{array}{cccc}
1+t_{3}+w_{3} & u_{3}+iv_{3} & t_{1}-it_{2}+w_{1}-iw_{2} & u_{1}-iu_{2}+iv_{1}+v_{2}\\
u_{3}-iv_{3} & 1+t_{3}-w_{3} & u_{1}-iu_{2}-iv_{1}-v_{2} & t_{1}-it_{2}-w_{1}+iw_{2}\\
t_{1}+it_{2}+w_{1}+iw_{2} & u_{1}+iu_{2}+iv_{1}-v_{2} & 1-t_{3}-w_{3} & -u_{3}-iv_{3}\\
u_{1}+iu_{2}-iv_{1}+v_{2} & t_{1}+it_{2}-w_{1}-iw_{2} & -u_{3}+iv_{3} & 1-t_{3}+w_{3}
\end{array}\right],
\end{equation}
\end{widetext}
Then, from identity (\ref{eq:jamshur}) follows that the action of (\ref{eq:maxEx1}) is to turn any map (\ref{eq:positive-matrix-general}) into:
\begin{equation}
\label{eq:prime-general}
    \Lambda^{\prime}=
    \begin{pmatrix}
       0 & 0 & (w_{1}+t_1)/2\\
0 & 0 & (w_{2}+t_2)/2\\
 0 & 0 & w_{3}
    \end{pmatrix},\quad
    \mathbf{t}^{\prime}=
    \begin{pmatrix}
        (w_{1}+t_1)/2\\(w_{2}+t_2)/2\\t_3
    \end{pmatrix},
\end{equation}
the above observation enables us to identify the
modification of parameters induced by (\ref{eq:maxEx1}), providing
knowledge of the underlying phenomenology.
Concretely, a recent experiment demonstrated the existence of dephasing
noise for qubit gates in a nuclear magnetic resonance (NMR) system \cite{Li2024}.
Using the experimental results of the previous reference we compute
the corresponding correlation matrix $\mathbf{C}^{(\textrm{exp})}$,
finding the following block matrices:
\begin{eqnarray}
C_{00}^{(\textrm{exp})} & = & \left(\begin{array}{cc}
1,000 & -0,066-0,368i\\
-0,066+0,368i & 1,000
\end{array}\right)\nonumber \\
C_{01}^{(\textrm{exp})} & = & \left(\begin{array}{cc}
0,003+0,465i & 0,701+0,000i\\
0,199+0,078i & -0,129+0,182i
\end{array}\right)\label{eq:ExpCorr}
\end{eqnarray}
Then, knowing $\mathbf{C}^{(\textrm{exp})}$, we can apply our characterization
methods, most notably finding $M(\Xi_{\mathbf{C}^{(\textrm{exp})}})=0.625>0$,
showing that it requires an active memory.
Consequently, $\mathbf{C}^{(\textrm{exp})}$ stands out for uncovering
dephasing noises\emph{ truly inherent to quantum gates}, distinct
from noise originating at the input or output systems and certifying
this novel phenomenon in real-world quantum computing architectures,
a crucial fact missing in the analysis of \cite{Li2024} but reported here
for the first time.

Additionally,  our methods allow us to investigate the impact of $\Xi_{\mathbf{C}^{(\textrm{exp})}}$ over the affine transformation \eqref{eq:positive-matrix-general}  and, after that, gain further insights of the underlying quantum stochastic processes by adapting noise-fitting methods \cite{Onorati2023}, an exploration deferred for future research.





\subsection{Non freely realizable operations}
\label{subsec:realisation}
It has been argued that the study of channel resources through freely realizable
operations is conceptually clearer and provides a natural partial
order between the channels under study \cite{Winter2019}. However, there
is no reason why all resource-non-generating superchannels should
be freely realizable. A compelling question is whether or not freely realizable
operations cover all physically implementable free operations in a
given channel resource theory, which remains open for most of the standard
theories  \cite{Liu2020,Carlo2021}.

An even more restrictive answer to the previous question is for freely realizable operations with a classical memory to encompass all physically meaningful free operations \cite{Li2020}. Such a conjecture is extremely strong and yet to date it remains open for theories as relevant as that of coherence generation. However,  our previous results allow us to conclusively answer in the negative the strong version of the conjecture.
Note, that taking $\rho$ incoherent in (\ref{eq:passive}), since $\N_{en}$ is MIO, the subsequent memory $\sigma$ is also incoherent. Because $\sigma$ is passive, the above holds for any $\rho$, and then the simulation must be equivalent to that of a classical memory i.e. $\fc=\sum_i q_{i} \mathrm{C}^{(i)}_1\otimes \mathrm{C}^{(i)}_2$.
Conversely, one can implement any classical memory simulation using a passive incoherent memory by writing the initial state as quantum-classical, using the classical system to correlate dephasing channels --which are MIO-- and then discarding the classical system. Our previous result (\ref{eq:diagonal-entries}) requires the above passive memory implementations to have a Gram matrix $\fc$ with identical diagonal elements  at every off-diagonal block matrix, constituting a measure zero set among all possible  dephasing superchannels, hence with infinitely many counterexamples, in particular Gram matrix (\ref{eq:maxEx1}).

We could finish our discussion about freely realizable dephasing superchannels with the above remark, nonetheless, by increasing the dimension of the target system, we found another family of counterexamples worth mentioning. In fact, consider the following family of dephasing superchannels $\Xi_{\mathbf{C}\left(\alpha,\beta\right)}$
which act on qutrit channels, with correlation matrices $\mathbf{C}\left(\alpha,\beta\right)$
given by:
\begin{align}
\mathbf{C}\left(\alpha,\beta\right)=\left(\begin{array}{ccc}
I & A\left(\alpha\right) & B\left(\beta\right)\\
A^{\dagger}\left(\alpha\right) & I & \mathbf{0}\\
B^{\dagger}\left(\beta\right) & \mathbf{0} & I
\end{array}\right)\label{eq:family}
\end{align}
with $I$ and $\mathbf{0}$ the identity and null matrices of dimension $3$,
and
\begin{eqnarray*}
A\left(\alpha\right)=\left(\begin{array}{ccc}
0 & 0 & 0\\
0 & 0 & 0\\
\alpha & 0 & 0
\end{array}\right) &, \ & B\left(\beta\right)=\left(\begin{array}{ccc}
\beta & 0 & 0\\
0 & 0 & 0\\
0 & 0 & 0
\end{array}\right),
\end{eqnarray*}
such that $\alpha,\beta$ are complex numbers in the unitary disk. Inside the above family  (\ref{eq:family}), the particular case $\alpha=\beta=1$ was presented in \cite{DephSuper2021} to show the non-simulability of general dephasing superchannels by pre- and post-processing of dephasing channels, which as we mentioned are equivalent to the freely realizable with passive memory. The crucial step in the mentioned proof is that the partial transpose of the Gram matrix $\mathbf{C}\left(1,1\right)$ has a negative eigenvalue and hence the Peres-Horodecki criterion guarantees that $\mathbf{C}\left(1, 1\right)$ is entangled. We can extend the same argument to the whole family $\mathbf{C}\left(\alpha,\beta\right)$, since their partial transpose has a smallest eigenvalue of $1-\sqrt{\left|\alpha\right|^{2}+\left|\beta\right|^{2}}$, which is negative iff $\left|\alpha\right|^{2}+\left|\beta\right|^{2}>1$.  While the above argument shows the impossibility of simulation due to entanglement, it omits some counterexamples revealed by our criteria.

Indeed, by invoking  Theorem \ref{thm:gram-pre-post2} and consequently Eq.\eqref{eq:diagonal-entries}, we directly find that those $\mathbf{C}\left(\alpha,\beta\right)$ satisfying $\left|\beta\right|^{2}>0$ define a superset of counterexamples to those witnessed by entanglement. The above shows that memory activity underlies the strength of the correlations in the Gram matrix. In this way, our Theorem \ref{thm:gram-pre-post2} allows us to clarify that memory activity rather than entanglement is the dominant factor in the non-simulation of noise.



Returning to the problem of the original conjecture:  Could we use MIO operations and an arbitrary active memory but initially incoherent to reproduce family $\mathbf{C}\left(\alpha,\beta\right)$? The answer is yes, and we show (Appendix~\ref{appD})  their implementation using an active memory and MIO; hence  we discard it as a source of counterexamples of the original conjecture. On the other hand, a research direction would be to exploit the formula (\ref{eq:gram-pre-post1})
and the analytical form of incoherent operations (IO) \cite{Chitambar2016} to make a numerical search of counterexamples in the qubit case, but by gradually increasing the dimension of the memory, which we leave for future research.

Nonetheless, we can observe that our main results unequivocally resolve the strong version of the simulation conjecture by freely realizable operations. The above underscores the tangible importance of Theorems \ref{thm:gram-pre-post} and \ref{thm:gram-pre-post2} while also pointing towards their potential to significantly contribute to resolving the original conjecture.

\section{Discussion}
\label{sec:discussion}
The inherent sensitivity of quantum effects to noise poses a significant challenge to harnessing quantum advantages for practical purposes. Consequently, understanding how noise affects quantum gates and systems is essential for designing robust quantum devices. In this work, we've addressed the previous challenge by extending a recent model of dephasing noise in quantum gates. Our main contribution is a formula that simplifies the task of determining if a set of operations and memory can simulate a specific dephasing noise.

One of the key aspects explored in our research is the role of memory activity in simulating dephasing superchannels. We have categorized memory into two types: passive and active. Surprisingly, our findings suggest that even in the case of qubit systems, active memories are necessary, challenging previous expectations \cite{DephSuper2021}. We also introduced a quantifier of memory activity required to simulate a dephasing noise and provided an analytical formula to compute it in the qubit case. Moreover, we apply our methods to analyze data from a recent experiment in a nuclear magnetic resonance (NMR) system \cite{Li2024} to certify the existence of an intrinsic gate's dephasing noise, irreducible to dephasing to either input or output states, for the very first time. The evidence unveiled by our methods provides a solid motivation for further study of the non-trivial noise manifestations of dephasing superchannels in the current quantum computing architectures \cite{metodi2011quantum,Bourassa2021blueprintscalable,Gullans2024}.


Our research also has significant implications for resource theories, which are essential for classifying and quantifying quantum resources. Specifically, we have examined the conjecture that any free operation within the resource theory of coherence generation can be implemented using memory and pre- and post-processing maximally incoherent operations (MIO). For the stronger version of this conjecture, which considers solely the use of classical memories, we have identified dephasing superchannels that serve as compelling counterexamples. This finding challenges existing assumptions and opens up new avenues for understanding free operations in the coherence generation resource theory.

As our research has highlighted, there is still much to explore in the realm of dephasing superchannels and their relevance for understanding memory and quantum operations. One intriguing research direction is investigating the physical meaning of correlations in the Gram matrices, specifying the dephasing noise. Despite memory activity overshadowing these correlation's role in determining non-trivial dephasing superchannels, a better explanation of the physical origin of these correlations is still needed. We expect that such an explanation could contribute to the profound question of the relationship between communication and memory, both associated with spatial and temporal correlations, respectively \cite{Brunner2014}.

While our research primarily focuses on the theoretical aspects of quantum information, the knowledge gained can have far-reaching consequences for the development of quantum technologies. Understanding the role of memory and the impact of noise on quantum operations is crucial for building more reliable and efficient quantum devices. Concretely, our work makes compelling the detection of non-trivial dephasing noise in gates on current architectures and the identification of appropriate error mitigation protocols for their new phenomenology.

In conclusion, our research into dephasing superchannels, memory activity, and their implications for resource theories represents a significant contribution to the field of quantum information. The unexpected findings challenge established assumptions and provide a fertile ground for further exploration. We are convinced that as quantum technologies continue to advance, the insights gained from our work will play a pivotal role in realizing the full potential of quantum devices.

\section*{Acknowledgement}
The authors would like to thank Carlo Maria Scandolo and Koorosh Sadri for their valuable comments and correspondence. RS acknowledges support by the Foundation for Polish Science (IRAP project, ICTQT, contract no.2018/MAB/5, co-financed by EU within Smart Growth Operational Programme). Financial support by the Foundation for Polish Science through TEAM-NET project (contract no. POIR.04.04.00-00-17C1/18-00) is acknowledged by FS.

\appendix

\section{Proof of Theorem~\ref{thm:gram-pre-post}}
\label{app:CBasedOnMaps}
Here we will show how the correlation matrix $\bc$ can be reconstructed knowing the encoding and decoding maps, $\N_{en}$ and $\N_{de}$, introduced in  Eq.~\eqref{eq:superchannel}. Let  $\Phi^{\mathsmaller{en}}$ and $\Phi^{\mathsmaller{de}}$ respectively denote the superoperator of  $\N_{en}$ and $\N_{de}$ and let $\tau$ be the initial state of the environment in the realization of a dephasing superchannel \eqref{eq:diagram1}. Moreover, let $\Phi$ denote the superoperator representation of the input quantum operation and note that to be acted on by the superchannel, it does not have to be trace-preserving. However, being a quantum operation means it is completely positive and its associated Jamio\l kowski state $J$, see~Eq.\eqref{eq:choi}, is positive.  We apply Eq.~\eqref{eq:jamshur} to form $\bc$. To do that, we employ Eq.~\eqref{eq:reshuffling} and the fact the superoperator representation for a concatenation of any two operations, $\E=\E_2\circ\E_1$, respects $\Phi^\E=\Phi^{\E_2}\Phi^{\E_1}$. Hence,
\begin{align}
    \left(J\odot\fc\right)_{\stackrel{\scriptstyle p i}{q j}}=
    \frac{1}{d}\sum_{\stackrel{\scriptstyle m n}{k l}}\sum_{\stackrel{\scriptstyle \theta,\gamma,\eta}{\alpha,\beta}}\Phi^{\mathsmaller{de}}_{\stackrel{\scriptstyle p\theta q\theta}{k\gamma l \eta }}\Phi_{\stackrel{\scriptstyle k l}{m n}}\Phi^{\mathsmaller{en}}_{\stackrel{\scriptstyle m\gamma n \eta}{i\alpha j \beta}}\ \tau_{\alpha\beta},
\end{align}
The above equality implies that for any $i,j,p,q$ and positive $J$
\begin{align}
    d J_{\stackrel{\scriptstyle p i}{q j}}\bc_{\stackrel{\scriptstyle p i}{q j}}=\sum_{\stackrel{\scriptstyle m n}{k l}}\sum_{\stackrel{\scriptstyle \theta,\gamma,\eta}{\alpha,\beta}}\Phi^{\mathsmaller{de}}_{\stackrel{\scriptstyle p\theta q\theta}{k\gamma l \eta }}\Phi_{\stackrel{\scriptstyle k l}{m n}}\Phi^{\mathsmaller{en}}_{\stackrel{\scriptstyle m\gamma n \eta}{i\alpha j \beta}}\ \tau_{\alpha\beta},
\end{align}
which in turn gives
\begin{align}
    &\sum_{\stackrel{\scriptstyle m n}{k l}} d J_{\stackrel{\scriptstyle km}{ln}}\times\\ \nonumber
    &\left(\bc_{\stackrel{\scriptstyle pi}{qj}}\delta_{kp}\delta_{mi}\delta_{lq}\delta_{jn}-\sum_{\stackrel{\scriptstyle \theta,\gamma,\eta}{\alpha,\beta}}\Phi^{\mathsmaller{de}}_{\stackrel{\scriptstyle p\theta q\theta}{k\gamma l \eta }}\Phi^{\mathsmaller{en}}_{\stackrel{\scriptstyle m\gamma n \eta}{i\alpha j \beta}}\ \tau_{\alpha\beta}\right)=0,
\end{align}
where we used Eq.~\eqref{eq:reshuffling} for $\Phi$. The above holds for any $J$ corresponding to a quantum operation. This means all positive matrices of dimension $d^2$ with a trace less than unity. The set of these matrices spans the entire set of matrices of the same dimension.
Next, we define $d^4$ matrices $B^{\stackrel{\scriptstyle pi}{qj}}$ for any $i,j,p,q$ as
\begin{align}
   B^{\stackrel{\scriptstyle  pi}{qj}}&=\bc_{\stackrel{\scriptstyle pi}{qj}}\ket{qj}\bra{pi}\\ \nonumber
   &-\sum_{\stackrel{\scriptstyle lm}{kn}}\left(\sum_{\stackrel{\scriptstyle \theta,\gamma,\eta}{\alpha,\beta}}\Phi^{\mathsmaller{de}}_{\stackrel{\scriptstyle p\theta q\theta}{k\gamma l \eta }}\Phi^{\mathsmaller{en}}_{\stackrel{\scriptstyle m\gamma n \eta}{i\alpha j \beta}}\ \tau_{\alpha\beta}\right)\ket{ln}\bra{km},
\end{align}
we get for any positive semi-definite  $J$
\begin{align}
    \Tr[JB^{\stackrel{\scriptstyle  pi}{qj}}]=0.
\end{align}
The above implies that for all $i,j,p,q$,  the matrices $B^{\stackrel{\scriptstyle  pi}{qj}}$ have to be zero which results in Eq.~\eqref{eq:gram-pre-post} and completes the proof of Theorem~\ref{thm:gram-pre-post}.

\section{Proof of Eq. \eqref{eq:N-pre-post}}
\label{app:dephasing-on-system}
In this Appendix, we bring a sketch of the proof of Eq.~\eqref{eq:N-pre-post}. The rigorous proof is in line with the proof of Theorem~\ref{thm:gram-pre-post2} which is a stronger version of Eq.~\eqref{eq:N-pre-post}.

Being idempotent under concatenation, the maximally dephasing channel $\D_I$ satisfies for all channels $\E$
\begin{equation}
\label{app-eq:classical}
    \D_I\circ\E\circ\D_I=\D_I\circ \D_I\circ\E\circ\D_I\circ\D_I.
\end{equation}
This means that the classical action of $\E$ and $\E'=\D_I\circ\E\circ\D_I$ are the same, see Eq.~\eqref{eq:classact}. On the other hand, a dephasing superchannel preserves the classical action of any channel by definition. Therefore, the classical actions of $\E$ and $\E'$ are equal with the left and the right following quantum circuits, respectively. This implies, due to Eq.~\eqref{app-eq:classical}, that these two circuits are equivalent. Thus, one can find equality, in the system's degree of freedom, of the block $1$ (highlighted in blue) in each circuit and the same for block $2$ (highlighted in red).
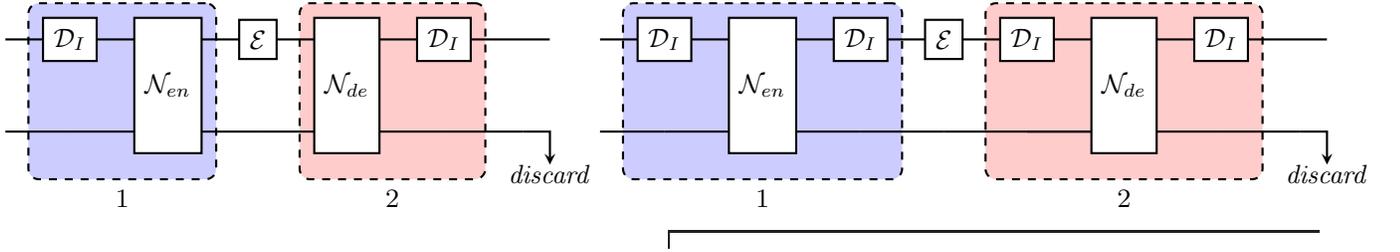
\begin{widetext}

\begin{quantikz}
    &\gate{\D_I}\gategroup[2,steps=2,style={dashed,rounded corners,fill=blue!20,inner sep=2pt},background,label style={label
position=below,anchor=north,yshift=-0.2cm}]{{\sc
1}}&\gate[2]{\N_{en}}&\gate{\E}&\gate[2]{\N_{de}}\gategroup[2,steps=2,style={dashed,rounded corners,fill=red!20,inner sep=2pt},background,label style={label
position=below,anchor=north,yshift=-0.2cm}]{{\sc
2}}&\gate{\D_I}&\qw \\
    & \qw & & \qw & & \qw & \trash{\text{\emph{discard}}}
\end{quantikz}
\begin{quantikz}
    &\gate{\D_I}\gategroup[2,steps=3,style={dashed,rounded corners,fill=blue!20,inner sep=2pt},background,label style={label
position=below,anchor=north,yshift=-0.2cm}]{{\sc
1}}&\gate[2]{\N_{en}}&\gate{\D_I}&\gate{\E}&\gate{\D_I}\gategroup[2,steps=3,style={dashed,rounded corners,fill=red!20,inner sep=2pt},background,label style={label
position=below,anchor=north,yshift=-0.2cm}]{{\sc
2}}&\gate[2]{\N_{de}}&\gate{\D_I}&\qw \\
    & \qw & & \qw&\qw &\qw & \qw &\qw &\trash{\text{\emph{discard}}}
\end{quantikz}

\end{widetext}

\section{Proof of Theorem ~\ref{thm:gram-pre-post2}}
\label{app:proof-of-T2}
Before proving this theorem, we emphasize that the action of the decoding map on those degrees of freedom of the memory that are not in the image of $\N_{en}$ plays no role.
\begin{proof}
We first prove if the superchannel is a dephasing superchannel, then Eq.~\eqref{eq:NS-enc-dec} has to hold. Later, we will prove Eq.~\eqref{eq:NS-enc-dec} is also sufficient.

To prove the first part, applying Theorem~\ref{thm:gram-pre-post}, let in Eq.~\eqref{eq:NS-enc-dec} $p=q$, $i=j$, $k=m$, and $l=n$. Thus,
\begin{align}
    \fc_{\stackrel{\scriptstyle im}{in}}\delta_{ip}=\sum_{\stackrel{\scriptstyle\theta,\gamma,\eta}{\alpha,\beta}}\Phi^{\mathsmaller{de}}_{\stackrel{\scriptstyle i\theta i\theta}{p\gamma p\eta}}\ \Phi^{\mathsmaller{en}}_{\stackrel{\scriptstyle m\gamma n\eta}{m\alpha n\beta}}\ \tau_{\alpha,\beta}.
\end{align}
Now, we sum over all indices $i$ in both sides of the above equation
\begin{align}
    \mathbf{C}_{\stackrel{\scriptstyle p m}{p n}}&=\sum_i\sum_{\stackrel{\scriptstyle\theta,\gamma,\eta}{\alpha,\beta}}\Phi^{\mathsmaller{de}}_{\stackrel{\scriptstyle i\theta i\theta}{p\gamma p\eta}}\ \Phi^{\mathsmaller{en}}_{\stackrel{\scriptstyle m\gamma n\eta}{m\alpha n\beta}}\ \tau_{\alpha,\beta}\nonumber\\[3pt]
    &=\sum_i\sum_{\theta,\gamma,\eta}\bra{i\theta}\N_{{de}}\left(\proj{p}\otimes\ketbra{\gamma}{\eta}\right)\ket{i\theta}\times\nonumber\\[-8 pt]
    &\qquad\qquad\ \ \ \bra{m\gamma}\N_{{en}}\left(\ketbra{m}{n}\otimes\tau\right)\ket{n\eta}\nonumber\\[3pt]
    &=\sum_{\gamma,\eta}\Tr\left[\N_{{de}}\left(\proj{p}\otimes\ketbra{\gamma}{\eta}\right)\right]\times\nonumber\\[-8 pt]
    &\qquad\quad \bra{m\gamma}\N_{{en}}\left(\ketbra{m}{n}\otimes\tau\right)\ket{m\eta}\nonumber\\[3pt]
    &=\sum_{\gamma,\eta}\Tr\left[\proj{p}\otimes\ketbra{\gamma}{\eta}\right]\times\nonumber\\[-8 pt]
    &\qquad\quad \bra{m\gamma}\N_{{en}}\left(\ketbra{m}{n}\otimes\tau\right)\ket{n\eta}\nonumber\\[3pt]
    &=\bra{m}\Tr_{2}\left[\N_{en}\left(\ketbra{m}{n}\otimes\tau\right)\right]\ket{n},
\end{align}
where the second equality is because of Eq.~\eqref{eq:superoperator}. Also, we applied the trace preserving property of the decoding channel to obtain the fourth equation. Applying the above equation, together with the fact that $\N_{en}$ and partial trace are two linear maps, shows that \eqref{subeq:enc-deph} must hold. Also, by noticing that the diagonal block of the correlation matrix is fixed, one gets that the left-hand side of the above equation is indeed independent of $p$ and that $\mathrm{C}_{en}=C_{00}$.

To prove \eqref{subeq:dec-deph}, we assume now in Eq.~\eqref{eq:gram-pre-post} that $k=l=m=n$, and we make use of Eq.~\eqref{subeq:enc-deph}
\begin{align}
    \mathbf{C}_{\stackrel{\scriptstyle i m}{j m}}\delta_{ip}\delta_{jq}&=\sum_{\stackrel{\scriptstyle\theta,\gamma,\eta}{\alpha,\beta}}\Phi^{\mathsmaller{de}}_{\stackrel{\scriptstyle i\theta j\theta}{p\gamma q\eta}}\ \Phi^{\mathsmaller{en}}_{\stackrel{\scriptstyle m\gamma m\eta}{m\alpha m\beta}}\ \tau_{\alpha,\beta}\\[3pt]
    &=\sum_{\theta,\gamma,\eta}\bra{i\theta}\N_{{de}}\left(\ketbra{p}{q}\otimes\ketbra{\gamma}{\eta}\right)\ket{j\theta}\times\nonumber\\[-8 pt]
    &\qquad\ \ \ \ \bra{m\gamma}\N_{{en}}\left(\ketbra{m}{m}\otimes\tau\right)\ket{m\eta}\nonumber\\[3pt]
    &=\sum_{\gamma,\eta}\bra{i}\Tr_2\left[\N_{{de}}\left(\ketbra{p}{q}\otimes\ketbra{\gamma}{\eta}\right)\right]\ket{j}\times\nonumber\\[-8 pt]
    &\qquad\quad \bra{m\gamma}\left(\proj{m}\otimes\sigma_{m}\right)\ket{m\eta}\nonumber\\[3pt]
   &=\bra{i}\Tr_2\left[\N_{{de}}\left(\ketbra{p}{q}\otimes\sigma_{m}\right)\right]\ket{j}.
\end{align}
Here, we applied the fact that $\N_{en}\left(\proj{m}\otimes\tau\right)=\proj{m}\otimes\sigma_{m}$ holds because the state $\proj{m}$ does not change under a dephasing channel and the only way for a bipartite state to have a pure marginal is to be product. Again, the above equation and the linearity of the maps $\N_{de}$ and $\Tr_2$ complete the proof of the necessity of  Eq.~\eqref{subeq:dec-deph} for the decoding map to give a dephasing superchannel.

To prove the inverse, we assume that Eq.~\eqref{eq:NS-enc-dec} holds and show the result will be a dephasing superchannel, i.e., it results in Eq.~\eqref{eq:invariant-classical}. In that order, note that for all $\rho$ and $\E$ it holds that
\begin{align}
    &\D_I\circ\Xi(\E)\circ\D_I(\rho)=\nonumber\\
    &(\D_I\otimes\Tr_2)\circ\N_{de}\circ(\E\otimes\I)\circ\N_{en}\left(\D_I(\rho)\otimes\tau\right)=\nonumber\\
    &\sum\rho_{mm}(\D_I\otimes\Tr_2)\circ\N_{de}\circ(\E\otimes\I)\circ\N_{en}\left(\proj{m}\otimes\tau\right)=\nonumber\\
    &\sum\rho_{mm}(\D_I\otimes\Tr_2)\circ\N_{de}\left(\E(\proj{m})\otimes\sigma_{m}\right)=\nonumber\\
    &\sum\rho_{mm}\D_I\circ\D_{\mathrm{C^m_{de}}}\left(\E(\proj{m})\right)=\nonumber\\
    &\D_I\circ\E\left(\sum\rho_{mm}\proj{m}\right)=\D_I\circ\E\circ\D_I(\rho),
\end{align}
where we used Eq.~\eqref{subeq:enc-deph} in the third inequality and Eq.~\eqref{subeq:dec-deph} in the forth. The fifth equality is obtained by employing the fact that the concatenation of any dephasing channel with the maximally dephasing one results in the latter, i.e., $\D_{\mathrm{C}}\circ\D_I=\D_I\circ\D_{\mathrm{C}}=\D_I$. The above shows that Eq.~\eqref{eq:NS-enc-dec} provides the sufficient condition to have a dephasing superchannel and completes the proof.
\end{proof}
\section{Evolution of off-diagonal elements for the case when the ancillary system plays the memory role}\label{app:off-diagonal}
To precisely show that even a passive memory affects system evolution, we use the Stinespring dilation of the map $\N_{en}$. To this end, let $U_{en}$ define the unitary operator in a Stinespring dilation of $\N_{en}$. Then, for two basis states $\ket{m}$ and $\ket{n}$ of the system, we get
\begin{subequations}
    \begin{align}
        U_{en}\ket{m\otimes\lambda\otimes0}&=\ket{m}\otimes\ket{\Psi_{\sigma^{\lambda,m}}}\nonumber\\
        &=\ket{m}\otimes\sum_{\alpha}\sqrt{p^{\lambda}_{\alpha}}\ket{\phi^{\lambda}_{\alpha}}\ket{\psi_{\alpha}^{\lambda,m}},\\
         U_{en}\ket{n\otimes\lambda\otimes0}&=\ket{n}\otimes\ket{\Psi_{\sigma^{\lambda,n}}}\nonumber\\
        &=\ket{n}\otimes\sum_{\alpha}\sqrt{p^{\lambda}_{\alpha}}\ket{\phi^{\lambda}_{\alpha}}\ket{\psi_{\alpha}^{\lambda,n}},
    \end{align}
\end{subequations}
where $\ket{0}\in \H_E$ is the environmental state needed in the Stinespring picture.  Moreover, we applied the fact that the basis states $\ket{i}$ for all $i$, and thus $i=m,n$, do not change. Therefore, they remain pure and have a product structure. Additionally, in the second equation in both equalities above, we wrote the state  $\ket{\psi_{\alpha}^{\lambda,i}}\in \H_D\otimes\H_E$ in the Schmidt basis applying the assumption that $\sigma^{\lambda,i}$ is $i$-independent. Thus the off-diagonal element $\ketbra{m}{n}$ evolves as
\begin{align}
    \N_{en}\left(\ketbra{m}{n}\right)&=\Tr_{DE}\left(U_{en}\ketbra{m\otimes\lambda\otimes0}{n\otimes\lambda\otimes0}U_{en}^\dagger\right)\nonumber\\
    &=\sum_\alpha p^{\lambda}_\alpha\bra{\psi_{\alpha}^{\lambda,n}}{\psi_{\alpha}^{\lambda,m}}\rangle \ketbra{m}{n}=C^\lambda_{mn}\ketbra{m}{n}.
\end{align}
The above shows how the ancillary system can affect the main system even once it has the role of memory with an invariant state.
\section{Explicit Physical Realization of $\mathbf{C}\left(\alpha,\beta\right)$}
\label{appD}
In this appendix, we will introduce the explicit form of controlled unitary gates in Eqs.~\eqref{eq:CU} and \eqref{eq:CV} by which we can realize the dephasing superchannel corresponding to  $\mathbf{C}\left(\alpha,\beta\right)$ introduced in \eqref{eq:family}. To this end, a set of states forming $\mathbf{C}\left(\alpha,\beta\right)$ will come to assist.
Consider $\ket{\psi_{ik}}:=V_iU_k\ket{00}$, then the Gram matrix reads $\mathbf{C}_{ik,jl}=\bra{\psi_{jl}}\psi_{ik}\rangle$. The set of states corresponding to $\mathbf{C}(\alpha,\beta)$ is defined as $\ket{\psi_{ik}}=\ket{ik}$ for all $i,k=\{1,2,3\}$ but $\ket{\psi_{21}}=\alpha^\ast\ket{13}+\sqrt{1-\abs\alpha^2}\ket{21}$ and $\ket{\psi_{31}}=\beta^\ast\ket{11}+\sqrt{1-\abs\beta^2}\ket{31}$.
Clearly, this set is not unique. Applying these states it is straightforward to find a set of unitary operations for the realization of $\mathbf{C}(\alpha,\beta)$. For a given set of states, these unitary operators are not unique since the only thing that plays a role is their action on the state $\ket{00}$. As an example of such a set of unitary operators applicable as a pre- and post-processing  to form the dephasing superchannel corresponding to $\mathbf{C}(\alpha,\beta)$ through Eqs.~\eqref{eq:CU} and \eqref{eq:CV}, we introduce $U_1=V_1=I$, $U_2=I\otimes \Pi_{12}$, $U_3=I\otimes\Pi_{13}$ with $\Pi_{mn}$ being the permutation of $m$ and $n$, and $V_2$, $V_3$ any unitary matrix with the following action
\begin{align*}
V_2\ket{11}&=\ket{\psi_{21}},\quad\quad V_3\ket{11}=\ket{\psi_{31}},\\
V_2\ket{12}&=\ket{\psi_{22}},\quad\quad V_3\ket{12}=\ket{\psi_{32}},\\
V_2\ket{13}&=\ket{\psi_{23}},\quad\quad V_3\ket{13}=\ket{\psi_{33}}.
\end{align*}

\bibliography{dephasing-superchannels}

\end{document}